\documentclass[letterpaper, 10 pt, conference]{IEEEconf}  

\usepackage[english]{babel}
\usepackage{hyperref}       
\usepackage{url,mathptmx}            
\usepackage{booktabs}       
\usepackage{amsfonts}       
\usepackage{nicefrac}       

\usepackage{microtype}      
\usepackage{lipsum}
\usepackage{amsthm}
\usepackage{amssymb}
\usepackage{amsmath}
\usepackage{amsfonts}
\usepackage{xcolor}
\usepackage[ruled,vlined]{algorithm2e}
\usepackage{subfigure}
\usepackage{mathtools}
\usepackage{comment}

\newtheorem{theorem}{Theorem}

\newtheorem{proposition}[theorem]{Proposition}

\newcommand{\ie}{\emph{i.e.}}

\IEEEoverridecommandlockouts                              
\usepackage{graphics} 

\title{\LARGE \bf
EPROACH: A Population Vaccination Game for Strategic Information Design to Enable Responsible COVID Reopening
}

\author{Shutian~Liu
        and~Quanyan~Zhu
\thanks{The authors are with the Department of Electrical and Computer Engineering, Tandon School of Engineering,
        New York University, Brooklyn, NY, 11201 USA (e-mail: sl6803@nyu.edu; qz494@nyu.edu).}       
       
}

\begin{document}

\maketitle
\thispagestyle{empty}
\pagestyle{empty}

\begin{abstract}
The COVID-19 lockdowns have created a significant socioeconomic impact on our society.
In this paper, we propose a population vaccination game framework, called EPROACH, to design policies for reopenings that guarantee post-opening public health safety.
In our framework, a population of players decides whether to vaccinate or not based on the public and private information they receive.
The reopening is captured by the switching of the game state.
The insights obtained from our framework include the appropriate vaccination coverage threshold for safe-reopening and information-based methods to incentivize individual vaccination decisions.
In particular, our framework bridges the modeling of the strategic behaviors of the populations and the spreading of infectious diseases.
This integration enables finding the threshold which guarantees a disease-free epidemic steady state under the population's Nash equilibrium vaccination decisions.
The equilibrium vaccination decisions depend on the information received by the agents.
It makes the steady-state epidemic severity controllable through information.
We find out that the externalities created by reopening lead to the coordination of the rational players in the population and result in a unique Nash equilibrium.
We use numerical experiments to corroborate the results and illustrate the design of public information for responsible reopening.

\end{abstract}

\section{Introduction}
%
%
%
%

\label{sec:intro}
The COVID-19 pandemic has created a significant socioeconomic impact on our society.
Due to the massive infections, many cities all over the world have been locked down.
Restrictions due to the lockdowns, including social distancing, pausing of restaurants and schools, and mask mandates, have been in place for nearly two years.
With the advent of vaccines, many cities are considering reopening plans.
One essential issue concerning the reopening policies is to reach a reasonable vaccination coverage rate before reopening the city so that the return to normal social activities does not generate new  outbreaks.
Therefore, it is important to find out the right vaccination coverage threshold for reopening and create careful reopening policies responsible for public health safety.
Once a threshold is determined, cities can focus on addressing the questions of the accompanying post-opening policies on mask-wearing and social distancing to reduce the risks of further outbreaks.
Apart from them, a concomitant question is to find ways to incentivize or nudge individuals to vaccinate.
The decision of the proper coverage threshold and the ways to reach this coverage are indispensable for the cities to prepare for reopening.

To this end, there are several challenges to decide the appropriate vaccination coverage threshold for reopening.
First, the threshold should take into account strategic decisions of human behaviors and guarantee effective herd immunity in the population. 
Second, the post-opening policies and the individual vaccination decisions are interdependent. 
Reopening of the city signals that the pandemic is over and makes people less in need of the mandates and restrictions.
The threshold or the reopening policies should take into account the human behaviors before and after the announced policies.

Challenges also arise when we aim to  reach the vaccination coverage threshold. 
Due to many reasons, people may choose not to vaccinate.
For example, people can be influenced by the misinformation from the media and believe that vaccination is harmful.
Some may be overly risk-averse toward uncertainties and doubt the safety of the vaccines.
People may also consider that the chance of them being infected is low when others are vaccinated. 
These facts make the independent vaccination choices of individuals hard to predict or control.
One possible solution to incentivize vaccinations is to provide monetary bonuses or gifts for taking vaccines.
The monetary incentives are helpful but they cannot scale in large populations.
Another option is to create   vaccination mandates.
However, this solution can lead to strong social opposition as we have witnessed in many countries.

\begin{figure}[ht]
\centering
\includegraphics[width=0.5\textwidth]{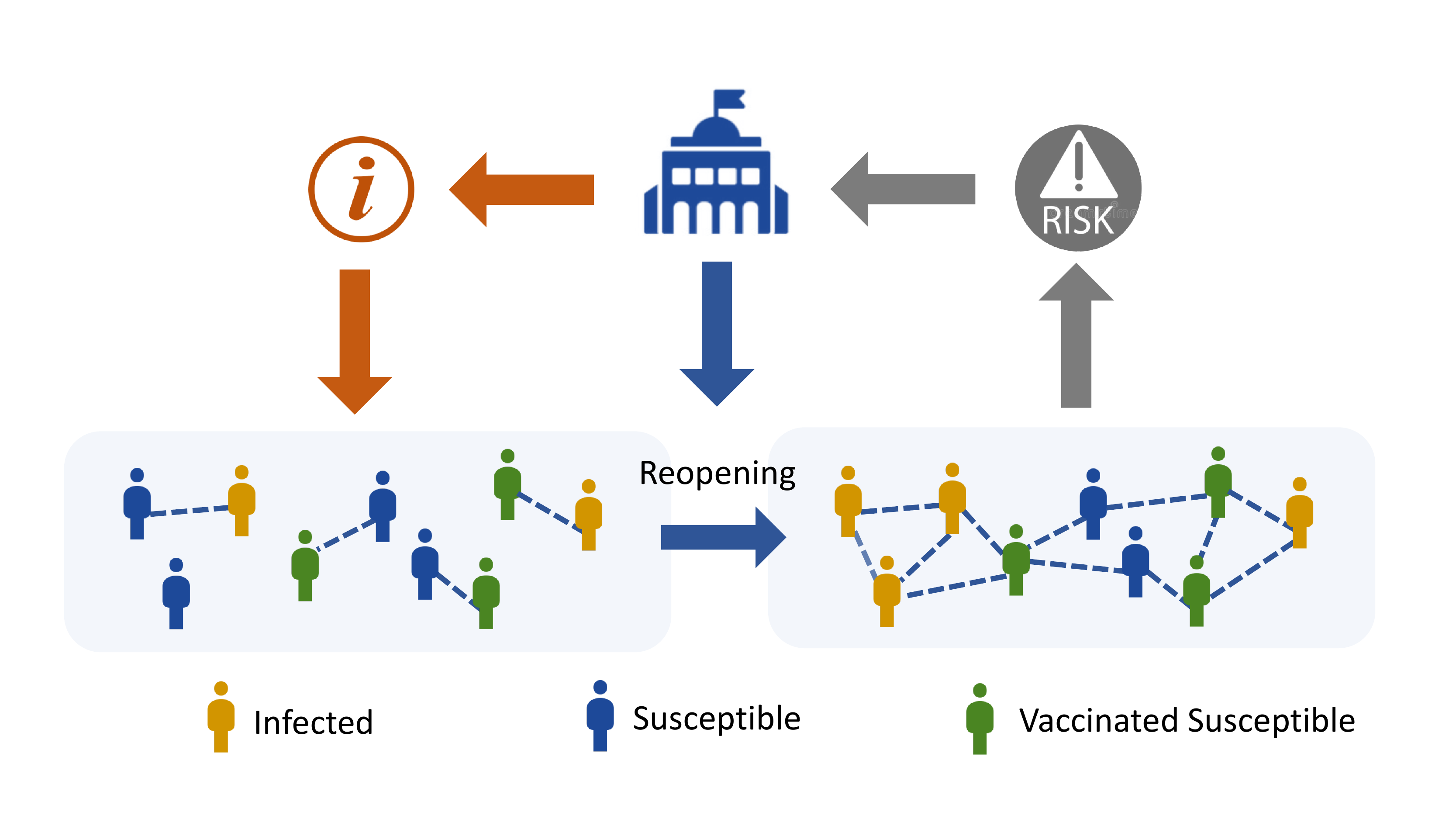}

\caption[Optional caption for list of figures]{During an epidemic, an authority determines whether or not to reopen a city.  The post-reopening virus outbreak risk is the key issue influencing this decision. The authority uses information to incentivize vaccination decisions of the population to reduce the risk.} 
\label{fig:concept}
\end{figure}

To address the above challenges, we propose a framework called the epidemiological reopening game of connected health (EPROACH).
EPROACH is a vaccination game framework integrated with the compartmental epidemic models.
We use a population game to capture the behavioral patterns resulting from individual decision-making.
In particular, we focus on vaccination decisions based on the well-being of individuals in the epidemic over a period of time. 
We incorporate public and private information into the game.
These two types of information shape individual beliefs of the severity of the epidemic and guide their vaccination decisions.
They also serve as the means to control the epidemic.
To capture the relations between the individual vaccination decisions and the effects of social policies before and after the reopening, we consider two regimes for the game.
These two regimes model the restricted and the reopened cities, respectively.
The regimes influence the population's behavior patterns through the spreading of the epidemic.
We use different compartmental epidemic models to quantify the risks of the outbreaks under various social policies adopted in the regimes.
These risks drive people's vaccination decisions and determine the accompanying social policies such as mask-wearing.
The direction of the evolution of the epidemic states eventually explains whether we arrive at herd immunity.
The integration of individual rational decision-making, private and public information, epidemic dynamics, and regime-switching in EPROACH helps us understand and quantify the outcomes of reopening and the potential risks attached. 
We propose incentivization strategies based on the Nash equilibrium (NE) of EPROACH to ensure public health safety after the reopening of cities.

We break the analysis of EPROACH into three parts.
The first part is about the externalities.
We analyze the incentives that drive people's vaccination decisions.
While the decisions seem selfish under the setting of either restricted or reopened regimes, they show a pattern of coordination when possible regime-switching can be triggered.
The incentives also generate multiple equilibria.
In the second part, we show the uniqueness of NE based on the results of the externalities.
The NE characterizes individual behaviors under the perceived infection risks in different regimes. 
The analysis of the NE provides a way to understand and predict behavioral patterns of the population under various reopening policies.
Apart from these, the relation between the information and the NE is also captured. 
The interdependence shows the impact of information on the outcomes of both vaccination coverage threshold and epidemic status.
The last part contributes to the design of information leveraging the connections between the information and the NE.
We call our method informational epidemic control since the virus can be eliminated by purely manipulating the public and private information.
Our method serves as a powerful economically friendly tool for incentivizing vaccinations since information is immaterial.
The vaccination coverage threshold required for safe-reopening is also obtained.
As a by-product, we provide advice on the vaccination cost.
These results on the designs all rely on the uncertainties of the information.
This aligns well with the fact that multiple equilibria are also generated by informational uncertainty.

We provide a brief literature review in the next section. 
The formulation of EPROACH is in Section \ref{sec: problem formulation}.
Section \ref{sec: analysis} contains the closely related elements of externalities and equilibrium analysis.
Leveraging the obtained results, we finally present the informational epidemic control method in Section \ref{sec: informational epidemic control}. 
The numerical experiments in Section \ref{sec:numerical experiments} contain further discussions on the results related to informational epidemic control.

\section{Related work}
\label{sec: related work}
One fundamental element of our framework is the compartmental epidemic model.
We use the degree-based mean-field model over complex networks \cite{pastor2015epidemic} to capture the contagion events within a networked population in an averaged sense.
Control of epidemics is one of the most investigated problems.
The recent \cite{acemoglu2021optimal} considers a balancing between virus testing and control of infections. 
The authors in \cite{chen2021optimal} consider epidemic control when there are competing viruses.
Their framework provides insights on curing policies of authorities when viruses have possibilities to mutate.
Epidemic processes are often studied together with human decision-making. 
Recently, \cite{liu2021herd} proposes a framework unifying individual decision-making processes and the epidemic dynamics to study herd behavior. 
The review \cite{huang2021game} summarizes and classifies popular game-theoretic models.
Our game-theoretic framework extends works emphasizing the interplay of human behavior and epidemics by focusing on public and private signals.
This enables information design.

Our framework builds on population games \cite{sandholm2010population}, which are suitable tools for studying behavioral patterns of groups of rational players.
By bringing public and private signals into population games, we can study herd behaviors in the setting of incomplete information.
This approach has been investigated in \cite{carlsson1993global} for understanding equilibrium selection when private signals influence the overall characteristics of the game. 
The monograph \cite{morris2001global} provides a clear and complete description of motivations, problem formulations, solution concepts, and application domains of these games.
Various settings have also been considered.
For example, \cite{frankel2003equilibrium} studies equilibrium selection under arbitrary number of players; \cite{harrison2021global} has studied the setting of submodular games; \cite{morris2004coordination,angeletos2007dynamic} have studied coordination under regime change.
Our framework makes the extension to complex networks to allow richer types of players and patterns of interactions.
Moreover, the underlying dynamical system captures players' coupled state transitions, adding additional constraints to the decision-making of players apart from the private and public signals.

\section{Formulation of EPROACH}
\label{sec: problem formulation}
Consider a population of heterogeneous players with a total mass $1$. 
These players have different degrees $d\in\mathcal{D}:=\{1,2,...,D\}$ which captures their intensities of interacting with other players.
The mass of players with degree $d$ is $m^d\in(0,1)$ with $\sum_{d\in\mathcal{D}}m^d=1$.
The population expands a complex network defined by the degree distribution $[m^d]_{d\in\mathcal{D}}$.
The average degree is denoted $\Bar{d}:=\sum_{d\in\mathcal{D}}dm^d$.
This population is exposed to the infections of a virus.
The players' probabilities of getting infected when they interact with each other over the network depend on their degree $d$.
Each individual  makes a one-time choice of whether or not to take the vaccination with effectiveness $\beta\in(0,1)$ designed for this virus at the beginning of the game.
The decisions are based on the anticipations of the players' own risks of getting infected at time $T>0$.
We call the period $(0,T)$ the period of interest.
The action set of players is $\mathcal{A}:=\{0,1\}$.
Action $i=0$ means that the player does not take the vaccination and action $i=1$ means that she takes the vaccination.
Note that the setting of binary action sets is without loss of generality. 
Our results in later sections can be generalized with more complicated notations to actions chosen from continuous sets that correspond to probabilities to vaccinate.

The regime of the game $s$ is chosen from the set $\{\text{'+','-'}\}$.
Regime $s=$'+' is the reopened state where there is no social restriction over the players.
Regime $s=$'-' is the restricted state where social activities that can cause contagions are restricted.
The restriction effect is denoted $\alpha^-=\alpha\in(0,1)$.
Correspondingly, we set $\alpha^+=1$ to denote the parameter of restriction-free social activities.
The switching of the regimes depends on two factors: the average action of players $A\in\mathbb{R}$ and the vaccination coverage threshold $\theta$, which is random with support $\mathbb{R}$.
Since $A$ captures the vaccination coverage rate, the state is '+' if $A\geq\theta$ and is '-' if $A<\theta$.

In this work, we view $\theta$ as the public signal of the game.
It is stochastic and drawn by nature.
While the distribution of $\theta$ is common knowledge, players perceive $\theta$ heterogeneously.
This heterogeneity is captured by the information type space $\mathcal{K}:=\{1,2,...,K\}$.
A player with information type $k\in\mathcal{K}$ receives a signal $x_k=\theta+\xi_k$, where $\xi_k$ is the perception bias of type $k$ and it is modeled as a randomness.
We call $x_k$ the private signal of type $k$ since the distribution of $\xi_k$ is only observable to players with information type $k$.  
The private signal is essential in shaping players' behaviors since it is correlated with the regime of the game $\theta$.

We assume that players with the same degree $d\in\mathcal{D}$ and the same information type $k\in\mathcal{K}$ are statistically equivalent.
This assumption is due to the following facts.
Firstly, the payoffs of the players depend on their risks of getting infected. 
This risk depends on a player's  intensity of interactions and is captured by the degree of a player.
Secondly, players' decisions are made based on their anticipations of the regimes.
This anticipation varies for players with different types, since beliefs of regimes are formed according to the correlations of the public signal $\theta$ and the private signals $x_k, k\in\mathcal{K}$.
Let $m^{d,k}_i$ denote the proportion of players with degree $d\in\mathcal{D}$ and type $k\in\mathcal{K}$ who select action $i\in\mathcal{A}$.
Let $m^d_i=\sum_{k\in\mathcal{K}}m^{d,k}_i$ and $m^k_i=\sum_{d\in\mathcal{D}}m^{d,k}_i$.
The average action can be expressed as $A=\sum_{d\in\mathcal{D}}\sum_{k\in\mathcal{K}}m^{d,k}_1=\sum_{d\in\mathcal{D}}m^d_1=\sum_{k\in\mathcal{K}}m^k_1$.
Note that when the subscript $i$ is dropped, we mean the players who play either $i=1$ or $i=0$.

Under our statistical equivalence assumption, we specify the payoffs of players under different regimes.
Let $c\in(0,1)$ denote the relative cost of vaccination.
Let $r\in\mathbb{R}$ denote the morbidity risk of the virus.
When the regime is $s=$'-', \ie, $A< \theta$,  a player  with degree $d$ and type $k$ observes payoff $u^{d,k,-}_0=-rI^{d,k,-}_{0}(T)$ under action $0$ and payoff $u^{d,k,-}_1=-c-rI^{d,k,-}_{1}(T)$ under action $1$, where $I^{d,k,-}_{i}(T)$ denotes the infection probability of her at the end of the period of interest when she plays action $i$ under regime $s=$'-'.
The infection probability $I^{d,k,-}_i$ in the population is governed by the epidemic process with recovery rate $\gamma\in(0,1)$ and contagion rate $\lambda\in(0,1)$ as follows:
\small
\begin{equation}
    \dot{I}^{d,k,-}_i(t)=-\gamma I^{d,k,-}_i(t)+\lambda_i(1-I^{d,k,-}_i(t))\alpha^- d \Theta^-(t),
    \label{eq:epidemic dynamics restricted}
\end{equation}
\normalsize
where $\lambda_i=\lambda$ if $i=0$ and $\lambda_i=\beta\lambda$ if $i=1$.
In (\ref{eq:epidemic dynamics restricted}), $\Theta^-(t):=\Bar{d}^{-1}\sum_{d\in\mathcal{D}}\sum_{k\in\mathcal{K}}\sum_{i\in\mathcal{A}}dm^{d,k}_iI^{d,k,-}_i(t)$ denotes the probability that a link is connected to an infected player at time $t$.
For the consistency of $\Theta^-(t)$, see \cite{pastor2001epidemic}. 
When the regime is $s=$'+', \ie, $A\geq\theta$, a player with degree $d$ and type $k$ observes payoff $u^{d,k,+}_0=-rI^{d,k,+}_{0}(T)+g^d$ under action $0$ and payoff $u^{d,k,+}_1=-c-rI^{d,k,+}_{1}(T)+g^d$ under action $1$, where $g^d$ denotes the utility gain a degree $d$ player generates after reopening.
Note that this utility gain can be related to psychological issues caused by quarantines and isolation.
Its dependence on $d$ characterizes the differences in the gains received by players having different degrees of social connections.
The term $I^{d,k,+}_i$ is the counterpart of $I^{d,k,-}_i$ under state $s=$'+'.
The corresponding epidemic process is
\small
\begin{equation}
    \dot{I}^{d,k,+}_i(t)=-\gamma I^{d,k,+}_i(t)+\lambda_i(1-I^{d,k,+}_i(t))\alpha^+ d \Theta^+(t),
    \label{eq:epidemic dynamics reopened}
\end{equation}
\normalsize
where $\Theta^+(t):=\Bar{d}^{-1}\sum_{d\in\mathcal{D}}\sum_{k\in\mathcal{K}}\sum_{i\in\mathcal{A}}dm^{d,k}_iI^{d,k,+}_i(t)$.
The initial infection probabilities are assumed to be increasing in players' degrees and independent of the types.
This is based on the practical concern that a player with a higher degree should have a higher initial infection probability.


\section{Equilibrium Analysis}
\label{sec: analysis}

In this section, we first analyze the structural property of (\ref{eq:epidemic dynamics restricted}) and (\ref{eq:epidemic dynamics reopened}).
The structural property helps us understand the incentives of the players.
This makes (\ref{eq:epidemic dynamics restricted}) and (\ref{eq:epidemic dynamics reopened}) more convenient to approach than they appear.
Then, we prove the uniqueness of the equilibrium based on these incentives.

We make the following assumptions on the distributions.
Assume that the vaccination coverage threshold $\theta$ follows a normal distribution with mean $\mu$ and variance $\frac{1}{\sigma^2}$, \ie, $\theta\sim \mathcal{N}(\mu,\frac{1}{\sigma^2})$. 
We set the nominal threshold represented by $\mu$ to be in $(0,1)$.
Among all the possible realizations of $\theta$ from the support $(-\infty, +\infty)$, the negative values mean that reopening does not rely on vaccinations; realizations that are higher than $1$ mean that no matter what the vaccination rate is, the city is always restricted.
The perception bias of type $k$ is assumed to have precision $\sigma^2_d$, \ie, $\xi_d\sim \mathcal{N}(0,\frac{1}{\sigma_d^2})$.
The results in this section also hold with general distributions \cite{morris2001global}.

\subsection{Incentive Analysis}
\label{sec: analysis: incentives}
Externality plays an important role in game-theoretic situations.
Network effects also generates externalities \cite{candogan2012optimal}.
Based on different communication structures, the externalities can be either global or local.
The systems of coupled differential equations (\ref{eq:epidemic dynamics restricted}) and (\ref{eq:epidemic dynamics reopened}) describing the epidemic evolution over the complex network are the sources of externalities of the players in our population game.
While the systems of equations (\ref{eq:epidemic dynamics restricted}) and (\ref{eq:epidemic dynamics reopened}) capture the local interactions of players in an averaged sense, the externalities of players show global patterns of incentives.

The following result shows that when social distancing is effective enough, the decisions in the game with potential regime switching are strategic complements. 

\begin{proposition}
\label{prop:supermodular game regime change}
Assume that the initial conditions of (\ref{eq:epidemic dynamics restricted}) and (\ref{eq:epidemic dynamics reopened}) are weakly increasing in players' degrees $d$ and independent of the types $k$.
Then, if the condition $\Theta^+(t)\geq \alpha \Theta^-(t)$ is satisfied for $t\in[0,T]$, the decisions in the population game with the possibility of regime-switching are strategic complements.
\end{proposition}
\begin{proof}
We show the strategic complements property by showing that the payoff gain of switching to action $1$ from action $0$ under state $s=$'+' is greater than that under state $s=$'-', \ie, $u^{d,k,+}_1-u^{d,k,+}_0 \geq u^{d,k,-}_1-u^{d,k,-}_0$, for all $d\in\mathcal{D}$ and $k\in\mathcal{K}$.
We obtain from (\ref{eq:epidemic dynamics restricted}) and (\ref{eq:epidemic dynamics reopened}) that
\small
\begin{equation}
    \begin{aligned}
        &\ \ u^{d,k,+}_1-u^{d,k,+}_0 \\
        & 
        =r(I^{d,k,+}_0(T)-I^{d,k,+}_1(T))-c \\
        &=r
        [ I^{d,k,+}_0(0)-I^{d,k,+}_1(0) + \int_0^{T} -\gamma(I^{d,k,+}_0(t)-I^{d,k,+}_1(t)) dt  \\ 
         &   \ \  +\lambda d \int_0^{T} \Theta^+(t)[(1-I^{d,k,+}_0(t))-\beta(1-I^{d,k,+}_1(t))]dt  ] -c .\\
        \label{eq:proof:u^+_1 - u^d_0}
    \end{aligned}
\end{equation}
\normalsize
Similarly, we obtain
\small
\begin{equation}
    \begin{aligned}
       & \ \ u^{d,k,-}_1-u^{d,k,-}_0 \\
        &= r
        [ I^{d,k,-}_0(0)-I^{d,k,-}_1(0) + \int_0^{T} -\gamma(I^{d,k,-}_0(t)-I^{d,k,-}_1(t)) dt  \\ 
         &   \ \  +\alpha \lambda d \int_0^{T} \Theta^+(t)[(1-I^{d,k,-}_0(t))-\beta(1-I^{d,k,-}_1(t))]dt  ] -c.\\
        \label{eq:proof:u^+_1 - u^d_0}
    \end{aligned}
\end{equation}
\normalsize
To compare the values of $u^{d,k,+}_1-u^{d,k,+}_0$ and $u^{d,k,-}_1-u^{d,k,-}_0$, consider the following difference of differential equations:
\small
\begin{equation*}
\begin{aligned}
        \Delta \dot{I}^{d,k,+} &=\dot{I}^{d,k,+}_0(t)-\dot{I}^{d,k,+}_1(t) \\
        &=-\gamma(I^{d,k,+}_0(t)-I^{d,k,+}_1(t)) \\
        & \ \ +  d \lambda \Theta^+(t) \left[(1-I^{d,k,+}_0(t))-\beta(1-I^{d,k,+}_1(t)) \right] ,
        \end{aligned}
\end{equation*}
\normalsize
and
\small
\begin{equation*}
    \begin{aligned}
        \Delta \dot{I}^{d,k,-} &=\dot{I}^{d,k,-}_0(t)-\dot{I}^{d,k,-}_1(t) \\
    &=-\gamma(I^{d,k,-}_0(t)-I^{d,k,-}_1(t)) \\
      &\ \ + \alpha d \lambda \Theta^-(t) \left[(1-I^{d,k,-}_0(t))-\beta(1-I^{d,k,-}_1(t)) \right].\\
    \end{aligned}
\end{equation*}
\normalsize
Under the condition $\Theta^+(t)\geq \alpha \Theta^-(t)$ for $t\in[0,T]$, $\Delta \dot{I}^{d,k,+}$ is an upper bound of $\Delta \dot{I}^{d,k,-}$.
According to \cite{brauer1963bounds}, the solutions of differential equations (\ref{eq:epidemic dynamics restricted}) and (\ref{eq:epidemic dynamics reopened}) satisfy $I^{d,k,+}_0(t)-I^{d,k,+}_1(t) \geq I^{d,k,-}_0(t)-I^{d,k,-}_1(t)$.
This implies that $u^{d,k,+}_1-u^{d,k,+}_0 \geq u^{d,k,-}_1-u^{d,k,-}_0$, which completes the proof.
\qed
\end{proof}
Note that the condition $\Theta^+(t)\geq \alpha \Theta^-(t)$ means that the effect of social restriction policy $\alpha$ makes the probability of linking to an infected player in the restricted state $s=$"-" lower than in the reopened state $s=$"+".

Strategic complementarity can be considered as a coordination of individual decisions.
It indicates that the public health connected through the complex network is coordinated.
A game is supermodular when the decisions of the players in this game are strategic complements \cite{topkis1979equilibrium}.
In a supermodular game, there are often multiple equilibria and simple evolutionary dynamics converge monotonically to these equilibria. 
Proposition \ref{prop:supermodular game regime change} provides the foundations for both analyzing the game from the perspective of players' incentives and solving the game using computational methods. 
One interesting point about Proposition \ref{prop:supermodular game regime change} is that the property of strategic complements does not depend on the utility gain $g^d$ of switching from $s=$'-' to $s=$'+'.
This utility gain originally makes the payoffs under the state $s=$'+' more attractive for the players.
However, it is the effectiveness of social policies, \ie, $\alpha$, which shapes players' decisions to become strategic complements.

Proposition \ref{prop:supermodular game regime change} is not a stand-alone result.
By following the similar reasoning as in the proof of Proposition \ref{prop:supermodular game regime change}, we arrive at the next result.

\begin{proposition}
\label{prop:submodular game fixed state}
Assume that the initial conditions of (\ref{eq:epidemic dynamics restricted}) and (\ref{eq:epidemic dynamics reopened}) are weakly increasing in players' degrees $d$ and independent of the types $k$.
If the switching of the regimes is absent, or, in other words, when the state is always either $s=$'+' or $s=$'-', the decisions in the population game are strategic substitutes.
\end{proposition}

The conflict of incentives shown in Proposition \ref{prop:supermodular game regime change} and Proposition \ref{prop:submodular game fixed state} arises from the possibility of the regime switching.
On the one hand, players' incentives to vaccinate exhibit rationality in the scenario of Proposition \ref{prop:submodular game fixed state}.
This rationality nudges individuals to protect themselves from getting infected without taking others' health conditions into consideration.
In particular, when the environment is more infectious due to the fact that less people decide to vaccinate, players become more eager to vaccinate to get protection from potential infections; when the environment is safer because of an ample amount of vaccinated, players become less motivated to vaccinate since they have less probability to get infected in daily lives.
On the other hand, players' incentives to vaccinate show a pattern of coordination in the scenario of Proposition \ref{prop:supermodular game regime change}.
The effect of this coordination is the congruence with the anticipations of the game.
In other words, the players tend to mimic others' decision on vaccination since unified actions either meet their expectation of switching to the reopened regime or decreases unnecessary losses except for those essential in maintaining the rationalities of themselves.

The coordination phenomenon found in Proposition \ref{prop:supermodular game regime change} implies a convenient structure of the incentives of the players when the outcome is captured by the Nash equilibrium.
Our equilibrium analysis in the next section leverages this structural property.

\subsection{Equilibrium Analysis}
\label{sec: analysis: NE}

Strategic complementarities generate multiple equilibria \cite{topkis1979equilibrium}.
The next result shows that a unique Nash equilibrium is selected when the public and private signals in our game satisfy certain conditions.
This equilibrium describes the vaccination decisions of the players when they observe their private signals.

\begin{proposition}
\label{prop:equilibrium of global game}
The global vaccination game admits a unique equilibrium in switching strategies if 
\begin{equation}
    \sum_{k\in\mathcal{K}}\frac{m^k}{\sigma_k}\leq \frac{\sqrt{2\pi}}{\sigma^2}.
    \label{eq:condition on uniqueness of equilibrium}
\end{equation}
\end{proposition}
\begin{proof}
Consider players with type $k$.
For a sufficiently low signal $x_k$, it is dominant for them to vaccinate because of the utility gains $g^d, \forall d\in\mathcal{D}$.
Similarly, it is dominant not to vaccinate for sufficiently high signals $x_k$.
Hence, it is natural to consider switching strategies with critical signals $x^*_k, \forall k\in\mathcal{K}$, \ie, players with type $k$ vaccinate if and only if $x_k \leq x^*_k$.

Given a vaccination coverage threshold $\theta$, the proportion of players with type $k$ who choose to vaccinate is $m^k_1=\sum_{d\in\mathcal{D}}m^{d,k}P(x_k\leq x^*_k|\theta)=m^kP(x_k\leq x^*_k|\theta)$, where the second equality follows the fact that the posterior probability for $x_k\leq x^*_k$ given $\theta$ is independent of players' degrees.
Let $\Phi$ denote the cdf of the standard normal, we obtain $P(x_k\leq x^*_k|\theta)=\Phi(\sigma_k(x^*_k-\theta))$.
Then, the population's averaged action, given $\theta$, is
$A(\theta)=\sum_{k\in\mathcal{K}}m^k_1=\sum_{k\in\mathcal{K}}m^k\Phi(\sigma_k(x^*_k-\theta))$.
Hence, the state of the game is '+' if and only if the threshold is low enough, \ie, $\theta\leq \theta^*$, where $\theta^*$ sovles
\begin{equation}
    \theta^*=A(\theta^*).
    \label{eq:proof:theta=A(theta)}
\end{equation}
For a player with type $k\in\mathcal{K}$, the posterior probability of regime change given her private signal $x_k$ is $P(\theta\leq \theta^*|x_k)$, which is decreasing in $x_k$.
Thus, she vaccinates if and only if 
$x_k\leq x^*_k$, where $x^*_k$ solves
\begin{equation}
    P(\theta\leq\theta^*|x^*_k)=c.
    \label{eq:proof:Posterior=c}
\end{equation}
Equation (\ref{eq:proof:Posterior=c}) means that the posterior probability of the state being '+' given critical signal compensates for the vaccination cost $c$.
Equation (\ref{eq:proof:Posterior=c}) can be explicitly expressed as
\begin{equation}
    \Phi\left( \sqrt{\sigma^2+\sigma^2_k}(\theta^*-\frac{\sigma^2_k}{\sigma^2+\sigma^2_k}x^*_k-\frac{\sigma^2}{\sigma^2+\sigma^2_k}\mu) \right)=c.
    \label{eq:proof:Posterior=c explicit}
\end{equation}

A switching strategy for players is given by $([x^*_k]_{\forall k\in\mathcal{K}},\theta^*)$, which solves (\ref{eq:proof:theta=A(theta)}) and (\ref{eq:proof:Posterior=c explicit}) jointly.

To show the existence and uniqueness of the solution to (\ref{eq:proof:theta=A(theta)}) and (\ref{eq:proof:Posterior=c explicit}), we first solve (\ref{eq:proof:Posterior=c explicit}) for $x^*_k$.
From (\ref{eq:proof:Posterior=c explicit}), we obtain
\begin{equation}
    x^*_k=\frac{\sqrt{\sigma^2+\sigma^2_k}\Phi^{-1}(c)+\sigma^2\mu-(\sigma^2+\sigma^2_k)\theta^*}{-\sigma^2_k}.
    \label{eq:proof:x*d explicit}
\end{equation}
Combining (\ref{eq:proof:x*d explicit}) and (\ref{eq:proof:theta=A(theta)}), we obtain
\small
\begin{equation}
\begin{aligned}
    \theta^*&=\sum_{k\in\mathcal{K}}m^k\Phi\left( \sigma_k(x^*_k-\theta^*) \right) \\
    &=\sum_{k\in\mathcal{K}}m^k\Phi\left( \frac{\sqrt{\sigma^2+\sigma^2_k}\Phi^{-1}(c)+\sigma^2\mu-\sigma^2\theta^*}{-\sigma_k} \right).
    \label{eq:proof:theta=W(theta)}
\end{aligned}
\end{equation}
\normalsize
Define $W:(0,1)\rightarrow \mathbb{R}$ by $W(\theta)=\sum_{k\in\mathcal{K}}m^k\Phi\left( \frac{\sqrt{\sigma^2+\sigma^2_k}\Phi^{-1}(c)+\sigma^2\mu-\sigma^2\theta}{-\sigma_k} \right)-\theta$.
The equality (\ref{eq:proof:theta=W(theta)}) can be written as $W(\theta^*)=0$.
To show the existence, we first note that $W$ is continuously differentiable. 
Moreover, the end points of $W(\cdot)$ satisfy
\small
\begin{equation*}
    \lim_{\theta \rightarrow 0}W(\theta)=\sum_{k\in\mathcal{K}}m^k\Phi\left( \frac{\sqrt{\sigma^2+\sigma^2_k}\Phi^{-1}(c)+\sigma^2\mu}{-\sigma_k} \right) > 0,
\end{equation*}
\normalsize
and
\small
\begin{equation}
    \lim_{\theta \rightarrow 1}W(\theta)=\sum_{k\in\mathcal{K}}m^k\Phi\left( \frac{\sqrt{\sigma^2+\sigma^2_k}\Phi^{-1}(c)+\sigma^2\mu-\sigma^2}{-\sigma_k} \right) -1 < 0,
    \label{eq:proof:limit theta goes to 1}
\end{equation}
\normalsize
where (\ref{eq:proof:limit theta goes to 1}) follows from the definition that $m^kP(x^k\leq x^*_k|\theta)$ represents a subset of players of the whole population.
Therefore, there exists a $\theta^*\in(0,1)$ which solves $W(\theta)=0$.

To show the uniqueness, we use the derivative of $W$:
\small
\begin{equation}
    \frac{\partial W(\theta)}{\partial \theta}= \sum_{k\in\mathcal{K}}m^k\phi(\frac{\sqrt{\sigma^2+\sigma^2_k\Phi^{-1}(c)}+\sigma^2\mu-\sigma^2\theta}{-\sigma_k})\frac{\sigma^2}{\sigma_k}  -1,
\end{equation}
\normalsize
where $\phi(\cdot)$ denotes the pdf of the standard normal.
Since $\lim_{\theta\rightarrow 0}W(\theta)>0$ and $\lim_{\theta\rightarrow 1}W(\theta)<0$, the uniqueness is guaranteed if $\frac{\partial W(\theta)}{\partial \theta}\leq0$.
When the argument is $0$, $\phi(\cdot)$ obtains the maximum value $\frac{1}{\sqrt{2\pi}}$.
Therefore, when (\ref{eq:condition on uniqueness of equilibrium})
is satisfied, we arrive at $\frac{\partial W(\theta)}{\partial \theta}\leq0$. 
This proves that $([x^*_k]_{\forall k\in\mathcal{K}},\theta^*)$ is the unique Nash equilibrium in switching strategies.

Next, we invoke the result in \cite{milgrom1990rationalizability} that there exists a unique strategy profile surviving iterated deletion of strictly dominant strategies for a game with strategic complementarity. 
The proof showing that the switching strategy  $([x^*_k]_{\forall k\in\mathcal{K}},\theta^*)$ survives iterated deletion of strictly dominant strategies follows the same arguments as \cite{morris2001global}.
Combining the strategic complementarity result from Proposition \ref{prop:supermodular game regime change}, we complete the proof.
\qed
\end{proof}

The reason why condition (\ref{eq:condition on uniqueness of equilibrium}) is independent of players' degrees is as follows.
Firstly, the structural property in Proposition \ref{prop:supermodular game regime change} allows us to focus on the behavioral patterns of the whole population rather than on strategy revisions of individuals.
This means that the strategic complementarity property has already taken into account the degree-dependent effects of infection risks captured by the coupled epidemic processes.
Secondly, the independence of degrees and types allows us to describe the behaviors of players with different degrees using the single posterior probability $P(x_k\leq x^*_k|\theta)$ as long as the players have the same type $k$.

We have incorporated many elements into our framework including switching regimes, public and private signals, and epidemic processes to capture the multifaceted behaviors of the players.
Under such complex settings, Proposition \ref{prop:equilibrium of global game}  shows that the players' vaccination decisions turn out to be predictable, \ie, the actions are captured by $([x^*_k]_{\forall k\in\mathcal{K}},\theta^*)$.

\section{Informational epidemic control}
\label{sec: informational epidemic control}

The long-term behaviors of the epidemic processes (\ref{eq:epidemic dynamics restricted}) and (\ref{eq:epidemic dynamics reopened}) are tightly connected with practical epidemic control policies. 
When a non-trivial steady state of the epidemics is considered, we often seek social policies which either decrease the total proportion of the infected or reduce the period required for reaching a desired level of the infected.
The reason lies in that the disease-free steady state often requires strong assumptions on the contagion rate and the recovery rate of the epidemics.
On the contrary, the integration of the epidemic spreading and the population game of incomplete information in our framework makes the disease-free steady state reachable through proper designs of the information.
In particular, we will show in the following that the public and private signals serve as tools to nudge people to vaccine.

We assume $\frac{\gamma}{d\lambda}\leq1, \forall d\in\mathcal{D}$. It is the condition that guarantees that the virus does not die out by itself \cite{liu2021herd}.

Let $(\Bar{\Theta}^+,\Bar{I}^+)$ denote the steady-state pair of $(\Theta^+(t),I^+(t))$ when the regime is '+'.
The pair $(\Bar{\Theta}^+,\Bar{I}^+)$ satisfies
\begin{equation}
    \Bar{\Theta}^+=\Bar{d}^{-1}\sum_{d\in\mathcal{D}}\sum_{k\in\mathcal{K}}\sum_{i\in\mathcal{A}}dm^{d,k}_i\Bar{I}^{d,k,+}_i,
    \label{eq:steady state eq 1 in global game epidemics}
\end{equation}
and
\begin{equation}
    \gamma\Bar{I}^{d,k,+}_i=\lambda_i(1-\Bar{I}^{d,k,+}_i)d\Bar{\Theta}^+.
    \label{eq:steady state eq 2 in global game epidemics}
\end{equation}
Equations (\ref{eq:steady state eq 1 in global game epidemics}) and (\ref{eq:steady state eq 2 in global game epidemics}) yield 
\begin{equation}
    \Bar{\Theta}^+=\Bar{d}^{-1}\sum_{d\in\mathcal{D}}\sum_{k\in\mathcal{K}}\sum_{i\in\mathcal{A}}dm^{d,k}_i\frac{\Bar{\Theta}^+}{\frac{\gamma}{\lambda_i d}-\Bar{\Theta}^+}.
    \label{eq:steady state Theta in global game epidemics}
\end{equation}
We observe from (\ref{eq:steady state Theta in global game epidemics}) that $(\Bar{\Theta}^+,\Bar{I}^+)=(0,0)$ is the disease-free steady-state pair.
The next result shows a sufficient condition under which a disease-free steady state in a ``reopened" regime is approachable through a proper design of players' accuracy of signals $\sigma_d, \forall d\in\mathcal{D}$.

\begin{proposition}
\label{prop:disease free steady state via private signal}
There exists $\Sigma^{d,k}\subset\mathbb{R}$, for all $d\in\mathcal{D}$ and all $k\in\mathcal{K}$, such that if $\sigma_k\in\cap_{d\in\mathcal{D}}\Sigma^{d,k}$, the disease-free steady-state pair $(\Bar{\Theta}^+,\Bar{I}^+)=(0,0)$ is globally asymptotically stable.

\end{proposition}
\begin{proof}
Since (\ref{eq:epidemic dynamics reopened}) is upper bounded by $-\gamma I^{d,k,+}_i(t)+\lambda_i d \Theta^+(t)$, it suffices to show the global asymptotic stability of the disease-free steady state $(\Bar{\Theta}^+,\Bar{I}^+)=(0,0)$ of the following auxiliary dynamical system:
\begin{equation*}
    \dot{I}^{d,k,+}_i=-\gamma I^{d,k,+}_i(t)+\lambda_i d \Theta^+(t).
\end{equation*}
To do so, we pick the Lyapunov function $V(t)=\sum_{d\in\mathcal{D}}\sum_{k\in\mathcal{K}}\sum_{i\in\mathcal{A}}b^{d,k}_iI^{d,k,+}_i(t)$, where $b^{d,k}_i=\frac{dm^{d,k}_i}{\gamma}\geq 0$.
Then, we obtain
\begin{equation*}
    \frac{dV(t)}{dt}=\Theta^+(t)\left( -\Bar{d}+\sum_{d\in\mathcal{D}}\sum_{k\in\mathcal{K}}\sum_{i\in\mathcal{A}}\frac{d^2\lambda_im^{d,k}_i}{\gamma} \right).
\end{equation*}
A sufficient condition to guarantee the global asymptotic stability is 
\begin{equation*}
    \sum_{d\in\mathcal{D}}\sum_{k\in\mathcal{K}}\sum_{i\in\mathcal{A}}\frac{d^2\lambda_im^{d,k}_i}{\gamma}\leq \Bar{d},
\end{equation*}
which is satisfied when
\begin{equation*}
    \frac{d\lambda}{\gamma}\sum_{k\in\mathcal{K}}(m^{d,k}_0+\beta m^{d,k}_1)\leq m^d, \forall d\in\mathcal{D},
\end{equation*}
or equivalently, when
\begin{equation}
    \sum_{k\in\mathcal{K}}m^{d,k}_1\geq \frac{m^d(\frac{\gamma}{d\lambda}-1)}{\beta -1}, \forall d\in\mathcal{D}.
    \label{eq:proof:condition on m^d_1}
\end{equation}
Inequality (\ref{eq:proof:condition on m^d_1}) is satisfied if
\begin{equation}
    m^{d,k}_1\geq\frac{m^{d,k}(\frac{\gamma}{d\lambda}-1)}{\beta -1}, \forall d\in\mathcal{D}, \forall k\in\mathcal{K}.
    \label{eq:proof:condition on m^dk_1}
\end{equation}
Observing that $x_k$ follows a normal distribution with mean $\mu$ and variance $\frac{1}{\sigma^2}+\frac{1}{\sigma^2_k}$, we obtain that the cumulative distribution function of $x_k$ is given by
\begin{equation*}
    \Phi_k(x_k):=\Phi\left( \frac{x_k-\mu}{\sqrt{\frac{\sigma^2+\sigma^2_k}{\sigma^2\sigma^2_k}}} \right).
\end{equation*}
According to the proof of Proposition \ref{prop:equilibrium of global game}, $m^{d,k}_1=m^{d,k}\Phi_k(x_k^*)$ where
\small
\begin{equation*}
    x^*_k=\frac{\sqrt{\sigma^2+\sigma^2_k}\Phi^{-1}(c)+\sigma^2\mu-(\sigma^2+\sigma^2_k)\theta^*}{-\sigma_k^2},
\end{equation*}
\normalsize
where $\theta^*$ satisfies the following fixed-point equation:
\small
\begin{equation}
    \theta^*=\sum_{k\in\mathcal{K}}m^k\Phi\left( \frac{\sqrt{\sigma^2+\sigma^2_k}\Phi^{-1}(c)+\sigma^2\mu-\sigma^2\theta^*}{-\sigma_k} \right).
    \label{eq:proof:fixed point eq of theta*}
\end{equation}
\normalsize
Therefore, combining (\ref{eq:proof:condition on m^dk_1}) with $\Phi_k(x_k^*)$ yields, for all $k\in\mathcal{K}$:
\small
\begin{equation}
    \theta^*\geq \hat{\theta}^*:= \mu+
    \frac{\Phi^{-1}(c)}{\sqrt{\sigma^2+\sigma^2_k}}+
    \frac{\sigma_k}{\sigma\sqrt{\sigma^2+\sigma^2_k}}\Phi^{-1}\left(e_d \right), 
    \label{eq:proof:condition on theta*}
\end{equation}
\normalsize
where $e_d:=\frac{\frac{\gamma}{d\lambda}-1}{\beta-1}$.
Next, we need to guarantee that $\theta^*$ solves the fixed-point equation (\ref{eq:proof:fixed point eq of theta*}) and satisfies the inequalities (\ref{eq:proof:condition on theta*}) simultaneously.
Recall from the proof of Proposition \ref{prop:equilibrium of global game} that $\theta^*$ solves the fixed-point equation (\ref{eq:proof:fixed point eq of theta*}) if and only if $W(\theta^*)=0$, and that $W(\cdot)$ is a continuous function and $\lim_{\theta\rightarrow 1}W(\theta)<0$.
Hence, it suffices to show that for $\theta^*=\hat{\theta}^*$,
$W(\theta^*)\geq 0$.
By combining the definition of $W(\theta)$ and equation (\ref{eq:proof:condition on theta*}), we obtain the following limit points for all $k\in\mathcal{K}$:
\begin{equation}
    \lim_{\sigma_k\rightarrow 0}W(\theta^*)=\sum_{k\in\mathcal{K}}m^ke_d-\mu-\frac{1}{\sigma}\Phi^{-1}(c),
    \label{eq:proof:limit at 0}
\end{equation}
and
\begin{equation}
    \lim_{\sigma_k\rightarrow \infty}W(\theta^*)=\sum_{k\in\mathcal{K}}m^k(1-c)-\mu-\frac{1}{\sigma}\Phi^{-1}(e_d),
    \label{eq:proof:limit at infty}
\end{equation}
Now, it suffices to show that (\ref{eq:proof:limit at 0}) and (\ref{eq:proof:limit at infty}) have the opposite sign.
We present the case where $ \lim_{\sigma_k\rightarrow 0}W(\theta^*)<0$ and $ \lim_{\sigma_k\rightarrow \infty}W(\theta^*)>0$.
The other case follows the same reasoning.
Suppose that $ \lim_{\sigma_k\rightarrow \infty}W(\theta^*)>0$, $\forall d\in\mathcal{D}$.
Then,
\begin{equation*}
    1-c-\mu-\frac{1}{\sigma}\Phi^{-1}(e_d)>0, \forall d\in\mathcal{D},
\end{equation*}
which is equivalent to 
\begin{equation}
    \Phi(\sigma(1-c-\mu))>e_d, \forall d\in\mathcal{D}.
    \label{eq:proof:condition on ed}
\end{equation}
Taking the weighted summation of (\ref{eq:proof:condition on ed}) over $k$, we obtain
\begin{equation*}
    \sum_{k\in\mathcal{K}}m^ke_d<\sum_{k\in\mathcal{K}}m^k\Phi(\sigma(1-c-\mu))=\Phi(\sigma(1-c-\mu)).
\end{equation*}
Then, the limit point can be represented as
\begin{equation*}
\begin{aligned}
    \lim_{\sigma_k\rightarrow 0}W(\theta^*)&=\sum_{k\in\mathcal{K}}m^de_d-\mu-\frac{1}{\sigma}\Phi^{-1}(c)\\
    &<\Phi(\sigma(1-c-\mu))-\mu-\frac{1}{\sigma}\Phi^{-1}(c).
    \end{aligned}
\end{equation*}
So, under condition $\Phi\left( \sigma(1-c-\mu) \right) < \mu +\frac{1}{\sigma}\Phi^{-1}(c)$.
The condition for the case where $\lim_{\sigma_d\rightarrow 0}W(\theta^*)>0$ and $ \lim_{\sigma_d\rightarrow \infty}W(\theta^*)<0$ hold is $\Phi\left( \sigma(1-c-\mu) \right) > \mu +\frac{1}{\sigma}\Phi^{-1}(c)$.
Therefore, we conclude that the disease-free steady state is globally asymptotically stable.
The existence of sets $\Sigma^{d,k}, \forall d\in\mathcal{D}, \forall k\in\mathcal{K}$ follows from the fact that $W(\theta^*)$ is continuous with respect to $\sigma_k$ and the limit points of $W(\theta^*)$ have opposite signs.
This completes the proof.
\qed
\end{proof}

Proposition \ref{prop:disease free steady state via private signal} has corroborated the existence of private signals that guarantee the global asymptotic stability of the disease-free steady state of the epidemic process under the reopened regime.
The private signals, when chosen from the sets $\cap_{d\in\mathcal{D}}\Sigma^{d,k}$, nudge players to take the vaccine and drive the epidemic to extinction. 
The vaccination coverage threshold which guarantees safe reopening is obtained in (\ref{eq:proof:condition on theta*}).

Next, we provide another sufficient condition for the stability of the disease-free steady state.

From the proof of Proposition \ref{prop:disease free steady state via private signal}, we know that a sufficient condition for the global asymptotic stability of $(\Bar{\Theta}^+,\Bar{I}^+)=(0,0)$ is $W(\hat{\theta}^*)\geq 0$.
From the definition of $W(\cdot)$, we obtain
\small
\begin{equation}
\begin{aligned}
    W(\hat{\theta}^*)&=\sum_{k\in\mathcal{K}}m^k\Phi\left( -\frac{\sigma_k}{\sqrt{\sigma^2+\sigma^2_k}}\Phi^{-1}(c)+\frac{\sigma}{\sqrt{\sigma^2+\sigma^2_k}}\Phi^{-1}(e_d) \right) \\
    & \ \ - \left( \mu+\frac{\Phi^{-1}(c)}{\sqrt{\sigma^2+\sigma^2_k}}+\frac{\sigma_k}{\sigma\sqrt{\sigma^2+\sigma^2_k}}\Phi^{-1}(e_d) \right).
    \label{eq:W(theta*)}
    \end{aligned}
\end{equation}
\normalsize
The first term on the right-hand side of (\ref{eq:W(theta*)}) is always positive.
Define $Y(\sigma_k):=\mu+\frac{\Phi^{-1}(c)}{\sqrt{\sigma^2+\sigma^2_k}}+\frac{\sigma_k}{\sigma\sqrt{\sigma^2+\sigma^2_k}}\Phi^{-1}(e_d)$.
If the condition $Y(\sigma_k)\leq 0$ is satisfied, the global asymptotic stability of the disease-free steady state is guaranteed.
Observing that $Y(\cdot)$ is continuously differentiable on $(0,+\infty)$, we obtain $Y(\sigma_k)\leq 0$ if and only if $\lim_{\sigma_k\rightarrow 0}Y(\sigma_k)\leq 0 $, $\lim_{\sigma_k\rightarrow +\infty}Y(\sigma_k)\leq 0 $, and $Y(\hat{\sigma}_k)\leq 0$, where $\hat{\sigma}_k=\sigma \frac{\Phi^{-1}(e_d)}{\Phi^{-1}(c)}$ solves $\frac{\partial Y(\sigma_k)}{\partial \sigma_k}=0$ being the stationary point of $Y(\cdot)$.
By expressing explicitly these three conditions, we arrive at the following result.
\begin{proposition}
\label{prop:disease free via public signal}
The disease-free steady state $(\Bar{\Theta}^+,\Bar{I}^+)=(0,0)$ is globally asymptotically stable if the following condition holds for all $d\in\mathcal{D}$:
\small
\begin{equation}
    \sigma \mu \leq \min \{-\Phi^{-1}(c), -\Phi^{-1}(e_d), -\frac{1+\Phi^{-1}(e_d)}{\sqrt{(\Phi^{-1}(c))^2+(\Phi^{-1}(e_d))^2}} \}.
    \label{eq:condition on sigma*mu}
\end{equation}
\normalsize
\end{proposition}

The product $\sigma\mu$ in (\ref{eq:condition on sigma*mu}) measures the concentration of the distribution of $\theta$. 
Its reciprocal $\frac{1}{\sigma\mu}$, called the coefficient of variation, measures the dispersion of a probability distribution.
Condition (\ref{eq:condition on sigma*mu}) requires the value of $\sigma\mu$ to be small. 
This indicates that it is the informational uncertainty about the vaccination coverage threshold which drives the epidemic process to the disease-free steady-state globally.

Since $\mu\in(0,1)$, the interesting case happens when the precision of the public signal goes to infinity, \ie, $\sigma\rightarrow +\infty$.
We leave the case with $\sigma\rightarrow +\infty$ and fixed $\mu$ and $\sigma_k$ in the numerical experiments. 
Now, suppose that the precision of the private signals also goes to infinity, but the ratio of the precisions of the public signal and the private signals satisfy (\ref{eq:condition on uniqueness of equilibrium}), \ie, $\sigma_k\rightarrow +\infty$, $\frac{\sigma^2}{\sigma_k}=l\leq \frac{\sqrt{2\pi}}{m_k}$.
Then, (\ref{eq:proof:fixed point eq of theta*}) becomes 
\begin{equation}
    \theta^*=\sum_{k\in\mathcal{K}}m^k\Phi\left( l\theta^*-l\mu+\Phi^{-1}(c) \right).
    \label{eq:fixed point eq of theta^* limit case}
\end{equation}
Leveraging the implicit function theorem, the change of $\theta^*$ with respect to $\mu$ can be expressed as:
\small
\begin{equation*}
\begin{aligned}
    \frac{\partial \theta^*}{\partial \mu}&=-\left(\sum_{k\in\mathcal{K}}m^k\phi(l\theta^*-l\mu+\Phi^{-1}(c))l-1 \right)^{-1} \\
    & \ \ \cdot \left( \sum_{k\in\mathcal{K}}m^k\phi(l\theta^*-l\mu+\Phi^{-1}(c))l \right)\geq 0.
    \end{aligned}
\end{equation*}
\normalsize
Hence, the solution to (\ref{eq:fixed point eq of theta^* limit case}) increases in $\mu$.
In this extreme scenario where $\sigma\rightarrow+\infty$, if (\ref{eq:condition on sigma*mu}) is satisfied, we need $\mu\rightarrow 0$ to hold.
Then, the vaccination coverage threshold $\theta^*$ solves
\begin{equation}
    \theta^*=\sum_{k\in\mathcal{K}}m^k\Phi\left( l\theta^*+\Phi^{-1}(c) \right),
    \label{eq:fixed point eq of theta^* sigma infty mu 0}
\end{equation}
which yields the minimum threshold when the public signal is infinitely precise and concentrates approximately at $0$.

Another perspective toward (\ref{eq:condition on sigma*mu}) is by focusing on the vaccination cost $c$, which is also a parameter that we can control.
From (\ref{eq:condition on sigma*mu}), we obtain the following inequality of $c$:
\small
\begin{equation}
    \Phi\left( -\sqrt{\frac{1+(\Phi^{-1}(e_d))^2}{(\sigma\mu)^2}-(\Phi^{-1}(e_d))^2} \right)\leq c\leq \Phi(-\sigma\mu).
    \label{eq:bounds on c}
\end{equation}
\normalsize
The upper bound of $c$ in (\ref{eq:bounds on c}) serves as a suggestion of pricing of the vaccination to the authorities. 
A high price of the vaccination  will weaken the motivations of people in taking the vaccination.
It results in a low vaccination coverage rate which cannot subdue the virus in the long run.
The lower bound of $c$ in (\ref{eq:bounds on c}) appears since (\ref{eq:condition on sigma*mu}) is only a sufficient condition.
The sufficiency arises from using the Lyapunov's method and the relaxation of the condition (\ref{eq:W(theta*)}).


\section{Numerical experiments}
\label{sec:numerical experiments}
In this section, we continue the discussion of Proposition \ref{prop:disease free via public signal} using numerical experiments. 
We provide a suggested region of the precision of the public information which guarantees the disease-free epidemic steady-state while maintaining a high probability of reopening under the NE. 
In the experiments, we study the effect of public information.
Hence, we set the degree of all players to be equal. 
We consider two information types.
The parameters $c, \mu, \sigma_k, \beta, e_d$ are chosen so that there is a unique NE and $-\Phi^{-1}(c)$ is the minimum element in the right-hand side of (\ref{eq:condition on sigma*mu}).

\begin{figure}[ht]
\centering
\includegraphics[width=0.5\textwidth]{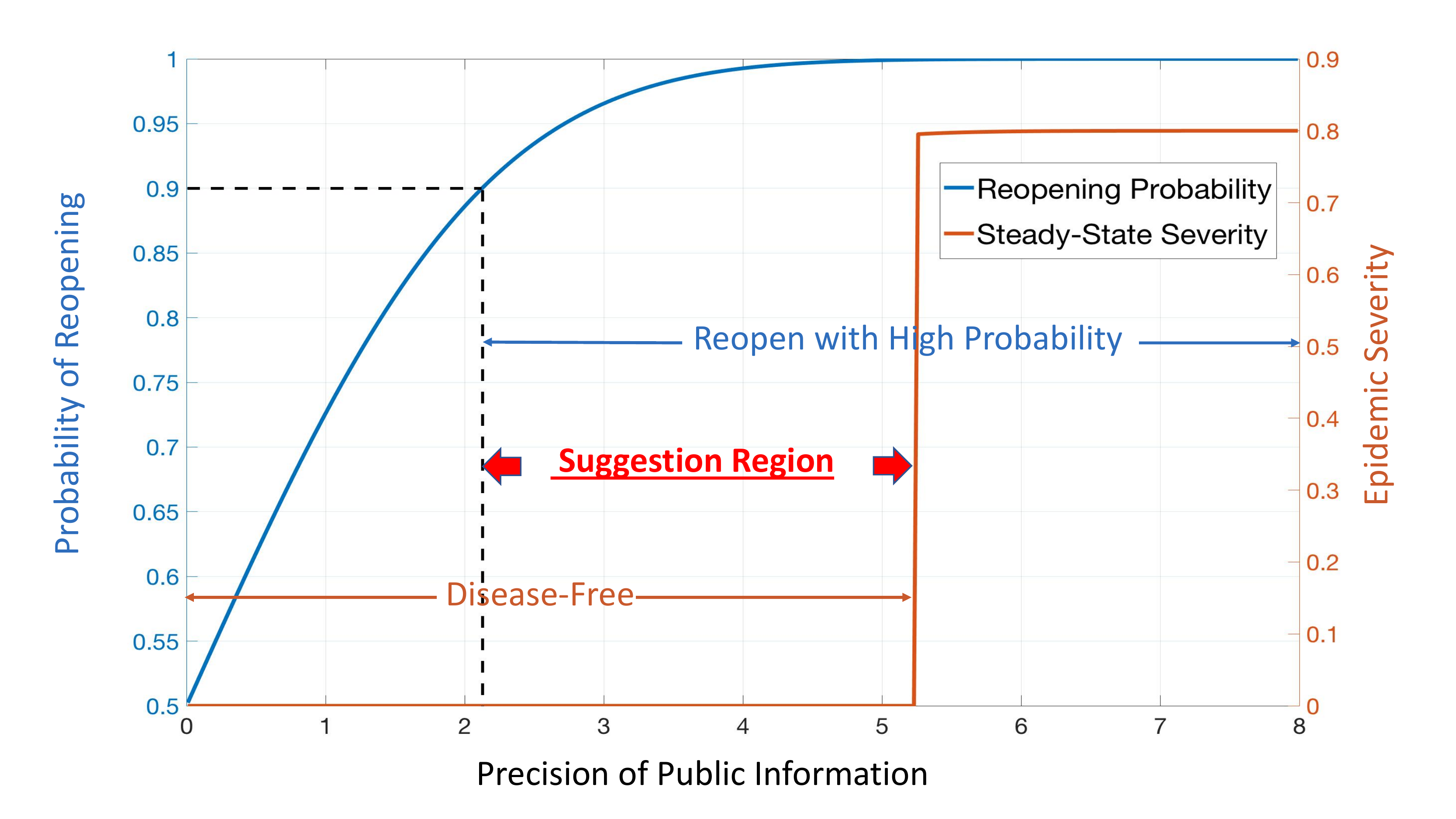}

\caption[Optional caption for list of figures]{The changes of the reopening probability and the steady-state epidemic severity when the precision of the public information varies are plotted. Let the desired probability of reopening be $0.9$ and we find the region of reopening with  high probability. Aiming to eliminate the virus, we find the region of disease-free steady state. The suggested region of the precision of the public information is the intersection of these two regions.} 
\label{fig:exp}
\end{figure}

In Fig. \ref{fig:exp}, we plot two indices for a given reopening plan.
The first index, the reopening probability, increases monotonically as the precision of public information becomes higher.
The second index, the epidemic severity at the steady state, is a piece-wise curve increasing in the precision of public information.
Suppose that the authority wants a high probability (chosen as $0.9$ in Fig. \ref{fig:exp}) for reopening so that the city has a higher chance to resume its activities.
Then, the desired public information lies above the value of the precision of public information which yields the reopening probability of $0.9$ (shown using the dashed lines).
Meanwhile, the way to control the outbreaks and the spreading of the virus affects the reopening plans.
Hence, the authority keeps the public information in the disease-free region.
As a consequence, our framework suggests a region of the precision of public information obtained at the intersection of the region of reopen with high probability and the disease-free region.
Within this suggested region, 
condition (\ref{eq:condition on sigma*mu}) is satisfied and the disease-free epidemic steady state is globally asymptotically stable.
The public information in this region is also precise enough in the sense that individuals, upon receiving their private signals, have a posteriori belief that the probability of regime-switching is high.
These beliefs will lead to affirmative vaccination decisions of a large proportion of the population, leading to the reopening of the city.

\section{Conclusion}
\label{sec:conclusion}
In this paper, we have proposed a vaccination game framework called EPROACH to capture the interplay between the population-level vaccination decisions under private and public information and the regime-switching of the epidemics. This framework has been motivated by developing reopening policies for cities when an increasing number of people become vaccinated.
Through the analysis of the externalities, we have found that the self-centered individual vaccination decisions become coordinated because of their anticipation to reopening.
Leveraging the coordination effect, we have found a unique Nash equilibrium of the game.
We have characterized players' vaccination decisions and the vaccination coverage threshold for safe reopening at this Nash equilibrium.
The uniqueness of the Nash equilibrium of our game has informed the design of reopening plans with the proposed informational epidemic control method.
We have studied the impact of informational design on vaccination decisions and the epidemic consequences before and after the reopening. 
We have observed that by choosing the resolution of the information, we can nudge the population to achieve the targeted vaccination threshold.
In the numerical experiments, the informational epidemic control method has provided a suggested region of the resolution of the public signal.
In this region, the reopening probability is high and the steady-state epidemic is eliminated after the reopening.

On possible extension of this work would consider the effects of the informational designs when players make dynamic decisions.
This requires detailed description of public and private history of observations for formulating players beliefs at different stages.
Another possible extension would consider peer-to-peer interactions in a finite network.
We would need to develop new techniques to characterize the incentives of the players, since the externalities may vary among players in different neighborhoods.


\bibliographystyle{IEEEtran}
\bibliography{bibliography.bib}

\begin{thebibliography}{10}
\providecommand{\url}[1]{#1}
\csname url@samestyle\endcsname
\providecommand{\newblock}{\relax}
\providecommand{\bibinfo}[2]{#2}
\providecommand{\BIBentrySTDinterwordspacing}{\spaceskip=0pt\relax}
\providecommand{\BIBentryALTinterwordstretchfactor}{4}
\providecommand{\BIBentryALTinterwordspacing}{\spaceskip=\fontdimen2\font plus
\BIBentryALTinterwordstretchfactor\fontdimen3\font minus
  \fontdimen4\font\relax}
\providecommand{\BIBforeignlanguage}[2]{{%
\expandafter\ifx\csname l@#1\endcsname\relax
\typeout{** WARNING: IEEEtran.bst: No hyphenation pattern has been}%
\typeout{** loaded for the language `#1'. Using the pattern for}%
\typeout{** the default language instead.}%
\else
\language=\csname l@#1\endcsname
\fi
#2}}
\providecommand{\BIBdecl}{\relax}
\BIBdecl

\bibitem{pastor2015epidemic}
R.~Pastor-Satorras, C.~Castellano, P.~Van~Mieghem, and A.~Vespignani,
  ``Epidemic processes in complex networks,'' \emph{Reviews of modern physics},
  vol.~87, no.~3, p. 925, 2015.

\bibitem{acemoglu2021optimal}
D.~Acemoglu, A.~Fallah, A.~Giometto, D.~Huttenlocher, A.~Ozdaglar, F.~Parise,
  and S.~Pattathil, ``Optimal adaptive testing for epidemic control: combining
  molecular and serology tests,'' \emph{arXiv preprint arXiv:2101.00773}, 2021.

\bibitem{chen2021optimal}
J.~Chen, Y.~Huang, R.~Zhang, and Q.~Zhu, ``Optimal curing strategy for
  competing epidemics spreading over complex networks,'' \emph{IEEE
  Transactions on Signal and Information Processing over Networks}, vol.~7, pp.
  294--308, 2021.

\bibitem{liu2021herd}
S.~Liu, Y.~Zhao, and Q.~Zhu, ``Herd behaviors in epidemics: A dynamics-coupled
  evolutionary games approach,'' \emph{arXiv preprint arXiv:2106.08998}, 2021.

\bibitem{huang2021game}
Y.~Huang and Q.~Zhu, ``Game-theoretic frameworks for epidemic spreading and
  human decision making: A review,'' \emph{arXiv preprint arXiv:2106.00214},
  2021.

\bibitem{sandholm2010population}
W.~H. Sandholm, \emph{Population games and evolutionary dynamics}.\hskip 1em
  plus 0.5em minus 0.4em\relax MIT press, 2010.

\bibitem{carlsson1993global}
H.~Carlsson and E.~Van~Damme, ``Global games and equilibrium selection,''
  \emph{Econometrica: Journal of the Econometric Society}, pp. 989--1018, 1993.

\bibitem{morris2001global}
S.~Morris and H.~S. Shin, ``Global games: Theory and applications,'' 2001.

\bibitem{frankel2003equilibrium}
D.~M. Frankel, S.~Morris, and A.~Pauzner, ``Equilibrium selection in global
  games with strategic complementarities,'' \emph{Journal of Economic Theory},
  vol. 108, no.~1, pp. 1--44, 2003.

\bibitem{harrison2021global}
R.~Harrison and P.~Jara-Moroni, ``Global games with strategic substitutes,''
  \emph{International Economic Review}, vol.~62, no.~1, pp. 141--173, 2021.

\bibitem{morris2004coordination}
S.~Morris and H.~S. Shin, ``Coordination risk and the price of debt,''
  \emph{European Economic Review}, vol.~48, no.~1, pp. 133--153, 2004.

\bibitem{angeletos2007dynamic}
G.-M. Angeletos, C.~Hellwig, and A.~Pavan, ``Dynamic global games of regime
  change: Learning, multiplicity, and the timing of attacks,''
  \emph{Econometrica}, vol.~75, no.~3, pp. 711--756, 2007.

\bibitem{pastor2001epidemic}
R.~Pastor-Satorras and A.~Vespignani, ``Epidemic spreading in scale-free
  networks,'' \emph{Physical review letters}, vol.~86, no.~14, p. 3200, 2001.

\bibitem{candogan2012optimal}
O.~Candogan, K.~Bimpikis, and A.~Ozdaglar, ``Optimal pricing in networks with
  externalities,'' \emph{Operations Research}, vol.~60, no.~4, pp. 883--905,
  2012.

\bibitem{brauer1963bounds}
F.~Brauer, ``Bounds for solutions of ordinary differential equations,''
  \emph{Proceedings of the American Mathematical Society}, vol.~14, no.~1, pp.
  36--43, 1963.

\bibitem{topkis1979equilibrium}
D.~M. Topkis, ``Equilibrium points in nonzero-sum n-person submodular games,''
  \emph{Siam Journal on control and optimization}, vol.~17, no.~6, pp.
  773--787, 1979.

\bibitem{milgrom1990rationalizability}
P.~Milgrom and J.~Roberts, ``Rationalizability, learning, and equilibrium in
  games with strategic complementarities,'' \emph{Econometrica: Journal of the
  Econometric Society}, pp. 1255--1277, 1990.

\bibitem{bauch2004vaccination}
C.~T. Bauch and D.~J. Earn, ``Vaccination and the theory of games,''
  \emph{Proceedings of the National Academy of Sciences}, vol. 101, no.~36, pp.
  13\,391--13\,394, 2004.

\bibitem{reluga2011general}
T.~C. Reluga and A.~P. Galvani, ``A general approach for population games with
  application to vaccination,'' \emph{Mathematical biosciences}, vol. 230,
  no.~2, pp. 67--78, 2011.

\end{thebibliography}
\nocite{*}

\end{document}